\documentclass[conference]{IEEEtran}
\usepackage{blindtext, graphicx}
\ifCLASSINFOpdf
\else
\fi
\usepackage[cmex10]{amsmath}
\usepackage{amssymb,amsthm}

\newtheorem{theorem}{Theorem}
\newtheorem{corollary}[theorem]{Corollary}
 
\newtheorem{definition}{Definition}
\newtheorem{proposition}[theorem]{Proposition}

\begin{document}

\setlength{\abovedisplayskip}{.2cm}
\setlength{\belowdisplayskip}{.2cm}

\title{Multipartite Monotones for Secure Sampling \\ by Public Discussion From Noisy Correlations\vspace{-1.1em}}  
\author{\IEEEauthorblockN{Pradeep Kr. Banerjee}
\IEEEauthorblockA{Indian Institute of Technology Kharagpur\\
Email: pradeep.banerjee@gmail.com}}
\maketitle

\begin{abstract}
We address the problem of quantifying the cryptographic content of probability distributions, in relation to an application to secure multi-party sampling against a passive \textit{t-adversary}. We generalize a recently introduced notion of assisted common information of a pair of correlated sources to that of \textit{K} sources and define a family of monotone rate regions indexed by \textit{K}. This allows for a simple characterization of all \textit{t-private} distributions that can be statistically securely sampled without any auxiliary setup of pre-shared noisy correlations. We also give a new monotone called the residual total correlation that admits a simple operational interpretation. Interestingly, for sampling with non-trivial setups (\textit{K} \textgreater \;2) in the public discussion model, our definition of a monotone region differs from the one by Prabhakaran and Prabhakaran (ITW 2012).
\end{abstract}


\IEEEpeerreviewmaketitle

\section{Introduction}
\vspace{2mm}
Suppose two parties, Alice and Bob working in distant labs have access to a certain set of nonlocal resources (e.g., noisy correlations or channels) and wish to simulate or realize the functionality of a target resource (e.g., oblivious transfer, a noiseless secret key, etc.).  Information-theoretic cryptography is concerned with the questions of \emph{feasibility} and \emph{efficiency} or rate of such reductions against computationally-unbounded adversaries. Given a set of $K$ parties, we focus on a restricted class of resources that takes no inputs from the parties, and following the execution of a distributed communication protocol over a public discussion channel, generates outputs $\{{Y_a}\}_{a=1}^K$ that approximately simulates a pre-specified joint distribution ${p_{{Y_1}, \ldots ,{Y_K}}}$. The protocol is required to be \emph{t-private}, i.e., any coalition of up to $t{\text{ }}( < K)$ honest-but-curious parties learns nothing more about the non-coalition parties' outputs than what they can derive from their own set of outputs. The problem is an instance of secure multi-party sampling (a form of secure multi-party computation with no inputs) that has recently gained a lot of currency in the information theory literature [1]--[4]. As a simple example, suppose Alice and Bob wish to sample pairs of the form, $((Y_1, Y_2): \Pr \{ Y_1 = Y_2\}  \ne \tfrac{1}{2})$. If they try to generate such a pair by talking to each other, they will necessarily end up violating 1-\emph{privacy}. On the other hand, pairs of the form ${Y_1}=({U_1},Q),\, {Y_2}=(Q,{U_2})$ where ${U_1}, {U_2}, Q$ are independent can be generated on the fly. However, outside this class of \emph{trivial} distributions, cryptographically useful non-trivial pairs $(Y_1, Y_2)$ cannot be securely realized from scratch, i.e., without the aid of an auxiliary \emph{setup} in the form of a trusted source of noisy correlations [1]--[3].

The earliest known impossibility result for secure 2-party sampling appears in the problem of mental poker [6]. Here two distant parties simulate the act of randomly sampling a disjoint pair of hands from a common deck of cards without using a trusted arbiter. Most relevant to the current work are the works on \emph{monotones}, real-valued functions of joint distributions that cannot increase under monopartite or local operations and noiseless public communication (LOPC). Monotones were first introduced in [5] as classical counterparts of entanglement or LOCC (local operations and classical communication) monotones to study the asymptotic rate of resource conversion under LOPC. Such rates are limited by the amount of resources contained in the source and target probability distributions. Monotones based on G{\'a}cs and K\"{o}rner's notion of the \emph{common part} of a pair of correlated sources [8] were introduced in [1] and later extended to the statistical case in [4]. Comparing the value of the monotone on the setup and protocol output random variables gives an upper bound on the rate of secure 2-party sampling.
Prabhakaran and Prabhakaran [2] developed a tighter upper bound technique using the concept of a \emph{monotone region} based on \emph{assisted common information}, a generalization of the G{\'a}cs-K\"{o}rner common information [8]. 
In [3], the same authors explored the power of different setups (or its lack thereof) in the multi-party scenario for different communication models, viz., the private channels model (parties linked via a complete network of bilateral secure channels) and the public discussion model.
A related work on the private channels model [7] gave a weak characterization of the class of \emph{t-private} distributions that are securely realizable from scratch, by reducing the problem to the 2-party case via a partition argument.

\emph{Contributions}. We address both the questions of feasibility and efficiency of statistically secure multi-party reductions in relation to sampling in the public discussion model. The main tool we develop is a generalization of the bivariate monotone region introduced in [2]. Our statistical impossibility result when specialized to the scenario of perfectly secure sampling from scratch, recovers the characterization in [3]. However, for the more general problem with non-trivial setups $(K > 2)$, our definition of a monotone region differs from the one in [3] and can give strictly better bounds on the rates of secure $K$-party protocols. We also give a new monotone called the residual total correlation that admits a simple operational interpretation.

\section{Preliminaries}
\vspace{2mm}
Random variables (RVs) and their finite alphabets are denoted using uppercase letters $X$ and script letters $\mathcal{X}$. We write $p_X$  to denote the distribution (pmf) of a discrete RV $X$. $X-Y-Z$ denotes that $X,Y,Z$ form a Markov chain satisfying ${p_{XYZ}} = {p_{XY}}{p_{Z|Y}}$. $A\backslash B$ denotes usual set-theoretic subtraction. The total variational distance between distributions $p_X$ and $p_{X'}$ is defined as $\mathsf{TV}(p_X, p_{X'}) \triangleq \tfrac{1}{2}{\lVert p_X-p_{X'}\lVert}_1$. For a nonnegative real coordinate space $\mathbb{R}_ + ^{d}$, the increasing hull of $\mathsf{A} \in \mathbb{R}_ + ^{d}$ is defined as $i(\mathsf{A}) \triangleq \{a \in \mathbb{R}_ + ^{d}: \exists a' \in \mathsf{A}$ s.t. $a \geq a'\}$ (where the comparison is coordinate-wise) [2].

For a pair $(X_1,X_2) \sim {p_{X_1X_2}}$, let $\mathcal{P}_{X_1,X_2}$ be the set of all RVs $Q$ jointly distributed with $(X_1,X_2)$. For ${p_{Q|X_1X_2}} \in \mathcal{P}_{X_1,X_2}$, $(X_1,X_2)$ is said to be \emph{perfectly resolvable} [2], if the residual information $I(X_1;X_2|Q)=0$, and $H(Q|X_1) = H(Q|X_2) = 0$. We then say that $Q$ \emph{perfectly resolves} $(X_1,X_2)$.

G{\'a}cs and K\"{o}rner (GK) [8] defined common information (CI) of the pair $(X_1,X_2) \sim {p_{X_1X_2}}$ as the maximum rate of common randomness (CR) that Alice and Bob, observing sequences $X_1^n$ and $X_2^n$ separately, can \emph{extract} without any communication. 
\begin{displaymath}
{C_{GK}}(X_1;X_2) \triangleq \mathop {\max }\limits_{\substack{Q:H(Q|X_1) = 0 \\ \hspace{3mm} H(Q|X2) = 0}} H(Q) = \mathop {\max }\limits_{\substack{ Q-X_1-X_2 \\ Q-X_2-X_1}} I(X_1X_2;Q).
\end{displaymath}
CR thus defined, is a far stronger resource than correlation, in that the latter does not result in common random bits, in general [8].  Nevertheless, when communication is an available resource, Alice and Bob can unlock hidden layers of potential CR. Following communication, the CR rate increases to $I(X_1;X_2)$.

Wyner [9] defined CI as the minimum rate of CR needed to \emph{generate} $X_1$ and $X_2$ separately using local operations (independent noisy channels: $Q \to X_1,Q \to X_2$) and no communication.
\begin{equation*}
{C_W}(X_1;X_2) \triangleq \mathop {\min }\limits_{Q:X_1 - Q - X_2} I(X_1X_2;Q),{\text{ }}|\mathcal{Q}| \leq |\mathcal{X}_1||\mathcal{X}_2|.
\end{equation*}
The three notions of CI are related as, ${C_{GK}}(X_1;X_2) \leq I(X_1;X_2) \leq {C_W}(X_1;X_2)$ with equality holding iff $(X_1,X_2)$ is perfectly resolvable, whence ${C_{GK}}(X_1;X_2) = I(X_1;X_2) \Leftrightarrow I(X_1;X_2) = {C_W}(X_1;X_2)$ [12].

\vspace{.5\baselineskip}
\emph{Common information duality in relation to the generalized Gray-Wyner Network}.
Consider the generalized Gray-Wyner (GW) distributed lossless source coding network [10], [11] shown in Fig. 1(a). The network jointly encodes $K$ discrete, memoryless correlated sources using a common message and $K$ private messages, and separately decodes each private message using the common message as side information. Let ${X_\mathcal{A}} \triangleq {\{ {X_a}\} _{a \in \mathcal{A}}}$ be a $K$-tuple of RVs ranging over finite sets $\mathcal{X}_a$ where $\mathcal{A}$ is an index set of size $K$.

\vspace{.5\baselineskip}
\begin{theorem}[\text{[11]}]
The optimal rate region ${\Re _{GW}}({X_\mathcal{A}})$ for the generalized GW network is given by
\begin{displaymath}
{\Re _{GW}}({X_\mathcal{A}}) = \begin{cases}
(\{{R_a}\} _{a = 1}^K,{R_0}) \in \mathbb{R}_ + ^{K + 1}:\exists {p_{Q|{X_\mathcal{A}}}} \in {{\hat{\mathcal{P}}}_{{X_\mathcal{A}}}},\\
\text{s.t. } {R_0} \geq I({X_\mathcal{A}};Q), \\ 
\hspace{5.5mm} {R_a} \geq H({X_a}|Q){\text{ }}\forall a \in \mathcal{A},\\
\end{cases}
\end{displaymath}
where ${\hat{\mathcal{P}}_{{X_\mathcal{A}}}}$ is the set of all conditional pmf's ${p_{Q|{X_\mathcal{A}}}}$ s.t. the cardinality of the alphabet $\mathcal{Q}$ of the auxiliary RV $Q$ is bounded as $|{\mathcal{Q}}| \leq \prod\nolimits_{a = 1}^K {|{\mathcal{X}_a}}| + 2$.
\end{theorem}

A trivial lower bound to ${\Re _{GW}}({X_\mathcal{A}})$ follows from basic information-theoretic considerations [10].
\begin{displaymath}
\begin{gathered}
  {\Re _{GW}}({X_\mathcal{A}}) \subseteq {\mathfrak{L}_{GW}}({X_\mathcal{A}}) \hfill \\
  \hspace{12mm} = \left\{ \begin{gathered}
  ({R_\mathcal{A}},{R_0}):{R_0} + {R_a} \geq H({X_a}){\text{ }}\forall a \in \mathcal{A}, \hfill \\
  \hspace{16mm} {R_0} + \sum\nolimits_{a = 1}^K {{R_a}}  \geq H({X_\mathcal{A}}) \hfill \\ 
\end{gathered}  \right\}. \hfill \\ 
\end{gathered} 
\end{displaymath}

Existing notions of CI can be viewed as extreme points for the corresponding common rate $R_0$ in the GW network (for $K=2$ see Problem 16.28--16.30, pg. 394 in [12]). For the generalized GW network, the CI duality is explicit when considering the complementary efficiency requirements of the first and second rate bundlings shown in Fig. 1(a). The inefficiency is manifest in the gap between ${\Re _{GW}}({X_\mathcal{A}})$ and the lower bound ${\mathfrak{L}_{GW}}({X_\mathcal{A}})$. 

When the sum-rate into each decoder (second bundling) is efficient (i.e., ${R_0} + {R_a} = H({X_a}),{\text{ }}\forall a \in {\mathcal{A}}$),  the maximum common rate is ${C_{GK}}({X_1}; \ldots ;{X_K})$ with the inefficiency in the first bundling being given by
\begin{align*}
  {\Delta _1} &= {R_0} + \sum\nolimits_{a = 1}^K {{R_a}}  - H({X_{\mathcal{A}}}) \hfill \\
    &= \sum\nolimits_{a = 1}^K {H({X_a}|Q)}  - H({X_{\mathcal{A}}}|Q)  = I({X_1}; \ldots ;{X_K}|Q)\tag{1}\label{K1-inefficiency} 
\end{align*} 
where the quantity, $I({X_1}; \ldots ;{X_K})$ is the total correlation [5] and is defined as $I({X_1}; \ldots ;{X_K})$ $\triangleq$ $\sum\nolimits_{a = 1}^K {H({X_a})}  - H({X_{\mathcal{A}}})$ $=$ $\sum\nolimits_{i = 1}^{K - 1} {I({X_1} \ldots {X_i};{X_{i + 1}})}$.
\begin{figure}[!t]
\centering
\includegraphics[width=2.5in]{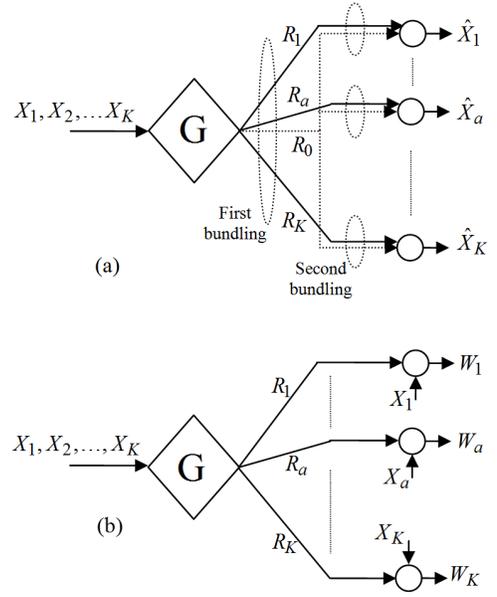}
\caption{(a) The generalized Gray-Wyner distributed source coding network (b) The generalized assisted common information setup}\vspace{-2.5em}
\label{fig_sim}
\end{figure}

When the sum-rate out of the $K$ encoders (first bundling) is efficient (i.e., ${R_0} + \sum\nolimits_{a = 1}^K {{R_a}}  = H({X_{\mathcal{A}}})$), the minimum common rate is ${C_{W}}({X_1}; \ldots ;{X_K})$  with the inefficiency in the second bundling being given by
\begin{align*}
  {\Delta _2} &= \sum\nolimits_{a = 1}^K ({{R_0} + {R_a} - H({X_a})})  \hfill \\
   &= \sum\nolimits_{a = 1}^K ({I({X_{\mathcal{A}}};Q)} + H({X_a}|Q) - H({X_a}))\\
   &\mathop=\limits^{{\text{(a)}}} \sum\nolimits_{a = 1}^K {I({X_{\mathcal{A}\backslash a}};Q|{X_a})}  = \sum\nolimits_{a = 1}^K {{\Delta _{2a}}} \tag{2}\label{K2-inefficiency}
\end{align*} 
where ${\Delta _{2a}} = I({X_{\mathcal{A}\backslash a}};Q|{X_a})$ captures the inefficiency of the $a$-th decoder and (a) follows from writing $I({X_{\mathcal{A}}};Q)$ as $I({X_{\mathcal{A}\backslash b}};Q) + I({X_b};Q|{X_{\mathcal{A}\backslash b}})$ $=$ $I({X_a};Q) + I({X_{\mathcal{A}\backslash ab}};Q|{X_a}) + I({X_b};Q|{X_{\mathcal{A}\backslash b}})$ $=$ $I({X_a};Q) + I({X_{\mathcal{A}\backslash a}};Q|{X_a})$. ${\Delta _1}$ and ${\Delta _2}$ are functions from ${\hat{\mathcal{P}}_{{X_\mathcal{A}}}} \to \mathbb{R}_ + ^{K + 1}$. In particular for $K=2$, the inefficiencies in the first and second bundlings are given by
\begin{align*}
  {\Delta _1} &= I({X_1};{X_2}|Q)  \\
  {\Delta _2} &= {\Delta _{21}} + {\Delta _{22}} = I({X_2};Q|{X_1}) + I({X_1};Q|{X_2})\tag{3}\label{2-inefficiency}
\end{align*}

Maximum efficiency of the first bundling occurs when ${\Delta _2} = 0$, i.e., $Q - {X_a} - {X_{\mathcal{A}\backslash a}},\forall a \in {\mathcal{A}}$. Similarly, maximum efficiency of the second bundling occurs when ${\Delta _1} = 0$, i.e., ${X_i} - Q - {X_j},{\text{ }}i \ne j,{\text{ }}\forall i,j \in {\mathcal{A}}$. It is easy to see that, 
\begin{align*}
\mathop {\min }\limits_{{\Delta _2} = 0} {\Delta _1} &= \mathop {\min }\limits_{Q - {X_a} - {X_{\mathcal{A}\backslash a}},\forall a \in {\mathcal{A}}} I({X_1}; \ldots ;{X_K}|Q)\\
&=I({X_1}; \ldots ;{X_K})-\mathop {\max }\limits_{\substack{Q - {X_a} - {X_{\mathcal{A}\backslash a}},\\ \hspace{3mm} \forall a \in {\mathcal{A}}}} I({X_1}; \ldots ;{X_K}|Q) \\
&=I({X_1}; \ldots ;{X_K}) - {C_{GK}}({X_1}; \ldots ;{X_K}) \tag{4}\\
\mathop {\min }\limits_{{\Delta _1} = 0} {\Delta _2} &= \mathop {\min }\limits_{\substack{{X_i} - Q - {X_j}, i \ne j, \\ \hspace{3mm} \forall i,j \in {\mathcal{A}}}} \sum\nolimits_{a = 1}^K {I({X_{\mathcal{A}\backslash a}};Q|{X_a})}\\
&= \mathop {\min }\limits_{\substack{{X_i} - Q - {X_j}, i \ne j, \\ \hspace{3mm} \forall i,j \in {\mathcal{A}}}} I({X_1} \ldots {X_K};Q) - I({X_1}; \ldots ;{X_K}) \\
&={C_W}({X_1}; \ldots ;{X_K}) - I({X_1}; \ldots ;{X_K}) \tag{5}
\end{align*}
Clearly, ${C_{GK}}({X_1}; \ldots ;{X_K}) = I({X_1}; \ldots ;{X_K}) \Leftrightarrow I({X_1}; \ldots ;{X_K})$ $= {C_W}({X_1}; \ldots ;{X_K})$.

It is interesting to note that, recently Prabhakaran and Prabhakaran [2] have introduced a rate region for a 3-party communication problem called the \emph{assisted residual information region}, ${\mathfrak{T}}({X_1};{X_2})$, which is the increasing hull of the set of all triples of the form $({\Delta _{21}}, \Delta _{22},\Delta _{1})$ $=$ $(I({X_2};Q|{X_1}),I({X_1};Q|{X_2}),I({X_1};{X_2}|Q))$. ${\mathfrak{T}}$ enjoys a certain monotonicity property lacking in the original GW region. From (3), it follows that ${\mathfrak{T}}({X_1};{X_2})$ is the image of ${\Re _{GW}}({X_1};{X_2})$ under an affine map that computes the inefficiencies of the first and second bundlings. Thus, ${\mathfrak{T}}({X_1};{X_2})$ formalizes the complementary efficiency requirements in terms of a rate-information trade-off region. 
Maximum efficiency occurs when ${\mathfrak{T}}({X_1};{X_2})$ includes the origin, which occurs when $({X_1},{X_2})$ is perfectly resolvable. At all other instances when the common core $Q$ fails to completely resolve the dependence between $({X_1},{X_2})$, ${\mathfrak{T}}({X_1};{X_2})$ is bounded away from the origin [2].

\section{Main contributions}
\vspace{2mm}
\subsection{The Generalized Assisted Residual Information Region}
Consider the setup in Fig. 1(b). Let ${X_\mathcal{A}} \triangleq {\{ {X_a}\} _{a \in \mathcal{A}}}$ be a $K$-tuple of RVs ranging over finite sets ${{\mathcal{X}}_a}$, where ${\mathcal{A}}$ is an index set of size $K$ and let $\{ {X_\mathcal{A}},i\} _{i = 1}^\infty$ be a sequence of independent copies ${X_{\mathcal{A},i}} \triangleq {\{ {X_{a,i}}\} _{a \in \mathcal{A}}}$ of ${X_\mathcal{A}}$ drawn i.i.d. $\sim {p_{{X_\mathcal{A}}}}$. $K$ terminals independently having access to one of the $K$ components of such a source are required to produce RVs ${\{ {W_a}\} _{a \in \mathcal{A}}}$ that must all agree with each other with high probability. An omniscient genie $\mathsf{G}$ having access to $X_{\mathcal{A}}^n$ assists the terminals by privately sending them rate-limited messages ${M_a} = f_a^n(X_\mathcal{A}^n)$, $a \in {\mathcal{A}}$ over noiseless links so that the terminals can independently compute ${W_a} = g_a^n(X_a^n,{M_a})$, $a \in {\mathcal{A}}$. We say that a $K$-tuple of rates $\{ {R_a}\} _{a = 1}^K$ \emph{enables residual information rate} ${R_0} \geq 0$ for ${X_\mathcal{A}}$ if for every $\epsilon > 0$ and $n$ sufficiently large, there exists deterministic mappings:
\[\begin{gathered}
  f_a^n:\mathcal{X}_1^n \times  \ldots  \times \mathcal{X}_K^n \to \{ 1, \ldots ,{2^{n({R_a} + \epsilon)}}\} ,{\text{ }}a \in \mathcal{A}, \hfill \\
  g_{a}^n:\mathcal{X}_a^n \times \{ 1, \ldots ,{2^{n({R_a} + \epsilon)}}\}  \to \mathbb{Z},{\text{ }}a \in \mathcal{A}, \hfill \\ 
\end{gathered} \]
where $\mathbb{Z}$ is the set of integers, s.t. $\forall i,j,a \in {\mathcal{A}}$
\[\begin{gathered}
  \Pr \{ g_i^n(X_i^n,{\text{ }}{M_i}) \ne g_j^n(X_j^n,{\text{ }}{M_j})\}  \leq \epsilon,{\text{ }}i \ne j, \hfill \\
  \tfrac{1}{n}I({X_1^n}; \ldots ;{X_K^n}|g_a^n(X_a^n,{\text{ }}{M_a})) \leq {R_0} + \epsilon. \hfill \\ 
\end{gathered} \]

\begin{definition}
The $(K+1)$-dimensional assisted residual information (ARI) rate region is defined as follows.
\[\begin{gathered}
{\mathfrak{T}}({X_{\mathcal{A}}}) \triangleq \{ (\{ {R_a}\} _{a = 1}^K,{R_0}):\{ {R_a}\} _{a = 1}^K \text{ enables residual} \\ 
\text{information rate } \mathop{R_0} \text{ for } {X_{\mathcal{A}}}\}.
\end{gathered}\]
\end{definition}

Denoting by ${\hat{\mathcal{P}}_{{X_\mathcal{A}}}}$ as the set of all conditional pmf's ${p_{Q|{X_\mathcal{A}}}}$ s.t. the cardinality of the alphabet $\mathcal{Q}$ of $Q$ is bounded as $|{\mathcal{Q}}| \leq \prod\nolimits_{a = 1}^K {|{\mathcal{X}_a}}| + 2$, the boundary of $\mathfrak{T}({X_{\mathcal{A}}})$ is made up of $(K+1)$-tuples of the form $\left( {{{\{ {\Delta _{2a}}\} }_{a \in \mathcal{A}}},{\Delta _1}} \right)$, and the rate region has the following characterization.
\begin{theorem} [Generalized $(K+1)$-dimensional assisted residual information region]
\begin{displaymath}
{\mathfrak{T}}({X_{\mathcal{A}}}) = \begin{cases}
(\{{R_a}\} _{a = 1}^K,{R_0}) \in \mathbb{R}_ + ^{K + 1}:\exists {p_{Q|{X_\mathcal{A}}}} \in {{\hat{\mathcal{P}}}_{{X_\mathcal{A}}}},\\
\text{s.t. } {R_0} \geq \sum\nolimits_{i = 1}^{K - 1} {I({X_1} \ldots {X_i};{X_{i + 1}}|Q)} = {\Delta _{1}}, \\ 
\hspace{5.5mm} {R_a} \geq I({X_{\mathcal{A}\backslash a}};Q|{X_a}) = {\Delta _{2a}}, {\text{ }}\forall a \in \mathcal{A}, \\
\end{cases}
\end{displaymath}
Also, ${\mathfrak{T}}({X_{\mathcal{A}}})$ is continuous, convex, and closed.
\end{theorem}

We sketch the proof of Theorem 2 in the Appendix. Corollary 3 follows from Theorem 2, (4), and (5) to yield the following expressions for the generalized G{\'a}cs-K\"{o}rner CI and Wyner CI in terms of the ARI region.
\begin{corollary}
\[\begin{aligned}
  {C_{GK}}({X_1};...;{X_K}) &= I({X_1};...;{X_K}) - \mathop {\min }\limits_{(0,...,0,{R_0}) \in \mathfrak{T}({X_\mathcal{A}})} {R_0}, \hfill \\
  {C_W}({X_1};...;{X_K}) &= I({X_1};...;{X_K}) + \mathop {\min }\limits_{\substack{({R_1},...,{R_K},0) \\ 
  \in \mathfrak{T}({X_\mathcal{A}})}} \sum\nolimits_{a = 1}^K {{R_a}} . \hfill \\ 
\end{aligned} \]
\end{corollary}

The following theorem (proven in the Appendix) gives the axes intercepts of the $(K+1)$-dimensional ARI region.
\begin{theorem}[Axes intercepts of the boundary of ${\mathfrak{T}}({X_{\mathcal{A}}})$]
\[\begin{gathered}
  \Delta _{2a}^{\operatorname{int} }({X_1}; \ldots ;{X_K}) \triangleq \min \left\{ {{R_a}:(0, \ldots ,{R_a}, \ldots ,0) \in \mathfrak{T}({X_{\mathcal{A}}})} \right\} \hfill \\
  {\text{            }} = \mathop {\min }\limits_{Q:{\text{ }}H(Q|{X_b}) = 0{\text{ }}\forall b{\text{ }} \in {\text{ }}{\mathcal{A}}\backslash a,{\text{ }}I({X_1}; \ldots ;{X_K}|Q) = 0} H(Q|{X_a}) \hfill \\
  \Delta _1^{\operatorname{int} }({X_1}; \ldots ;{X_K}) \triangleq \min \left\{ {{R_0}:(0, \ldots ,0,{R_0}) \in \mathfrak{T}({X_{\mathcal{A}}})} \right\} \hfill \\
  {\text{            }} = \mathop {\min }\limits_{Q:{\text{ }}H(Q|{X_a}) = 0{\text{ }}\forall a \in {\mathcal{A}}} I({X_1}; \ldots ;{X_K}|Q) \hfill \\ 
\end{gathered} \]
\end{theorem}

\subsection{Monotone Regions for Secure K-party Sampling with Public Discussion}

We establish the monotonicity properties of $\mathfrak{T}$, which by virtue of being continuous and convex allows for deriving tight outer bounds on the rate of statistically secure sampling for the general $K$-party problem with setups. It is well-known that cryptographically useful non-trivial distributions cannot be securely realized from scratch, i.e., without the aid of an auxiliary setup of correlated randomness [3], [7]. Trusted pre-shared noisy correlations is a simple yet powerful cryptographic resource that takes no inputs from the parties, and generates samples of a given joint distribution, with party-$i$ given access to copies of the $i$-th variable. Access to such a setup is known to realize 2-party sampling [2], as well as other important primitives like bit commitment and oblivious transfer [1], [4] in an unconditionally secure way. In light of the resource character of noisy correlations in enabling such reductions (which are otherwise impossible to realize from scratch), abstracting and quantifying such resources is of interest. A resource is specified by a restriction, $\mathfrak{C}$ on the full set of realizable operations. Given $\mathfrak{C}$, states that cannot be created by means of $\mathfrak{C}$ naturally acquire some value and become a resource. When distant parties wish to securely sample RVs by manipulating a given joint distribution, it is natural to restrict attention to the class of LOPC operations. The resourcefulness or cryptographic content of the distribution is a nonlocal property that cannot increase under LOPC, and can be quantified using monotones. Monotones for secure $K$-party sampling are real-valued quantities that can never increase in any protocol that securely realizes a $K$-tuple of correlated RVs ${Y_{\mathcal{A}}}$ using a setup ${X_\mathcal{A}}$. As we shall see, the entire region $\mathfrak{T}$ is a monotone and $\mathfrak{T}({Y_{\mathcal{A}}})$ can be interpreted as a \emph{witness} of the cryptographically trivial nature of ${Y_{\mathcal{A}}}$: ${Y_{\mathcal{A}}}$ can be perfectly securely realized from scratch, iff $\mathfrak{T}({Y_{\mathcal{A}}})$ contains the origin. The closer ${Y_{\mathcal{A}}}$ is to the origin, the lesser cryptographic content it has. Conversely, the lesser $\mathfrak{T}({Y_{\mathcal{A}}})$ bulges towards the origin, the more cryptographic content it has.

Consider the following simplified description of the \emph{semi-honest} model for secure $K$-party sampling [3], [7]. 
A set of $K$ parties engage in an interactive (randomized) communication protocol $\Pi$ over a public discussion channel to accomplish the distributed approximate simulation of a prescribed joint distribution ${p_{{Y_\mathcal{A}}}}$. The parties have access to an auxiliary setup: independent copies of jointly distributed RVs ${X_\mathcal{A}}\sim {p_{X_\mathcal{A}}}$, with party-$a$ independently having access to copies of $X_a$ as well as an infinite stream of private randomness. The protocol proceeds in rounds, where in each round each party flips private coins, and based on the messages exchanged so far, sends a message over a broadcast public communication channel to all the other parties. At the end of the protocol, party-$a$ generates output ${\hat Y_a}$ as a function of its \emph{view} (encapsulated in the RV $V_a$), which consists of copies of its setup RV $X_a$, all the private coins flipped so far, and all the communication received over all the previous rounds. Interfering with the interaction is a semi-honest \emph{t-adversary} who may choose to ``passively corrupt'' a set $\mathcal{T}$ ($\subset\mathcal{A}$) of at most  $t$ ($< K$) parties, and learn their internal states. Compared to perfect reductions, statistical implementations are much more efficient [4]. The privacy and correctness requirements [3], [7] for statistically secure reductions can be stated as follows.

\begin{definition}
For $\epsilon$, $\delta \geq 0$, a protocol $\Pi$ is ($\delta$,t)-private if the information leakage of the final views of the corrupted parties $({V_\mathcal{T}})$ satisfies 
\begin{displaymath}
\sum\limits_{\mathcal{T} \subset \mathcal{A}:\lvert{\mathcal{T}}\rvert \leq t} {I({V_\mathcal{T}};{{\hat Y}_{\mathcal{A}\backslash \mathcal{T}}}|{{\hat Y}_\mathcal{T}}) \leq \delta}.
\end{displaymath}
The protocol is $\epsilon$-correct if $\mathsf{TV}({p_{{Y_\mathcal{A}}}},{p_{{{\hat Y}_\mathcal{A}}}})\leq\epsilon$. Perfect privacy and correctness correspond to $\delta=0$ and $\epsilon=0$, respectively.
\end{definition}

$(\delta,t)$-\emph{privacy} implies that any coalition of up to $t{\text{ }}( < K)$ parties who are honest but ``curious'' and leak their entire final views, learns nothing more about the non-coalition parties' outputs than what they can derive from their own set of outputs. 
As the views of the parties evolve along any LOPC protocol, the region of residual total dependency of the views can never shrink (away from the origin) [2].
Thus, if $\Pi$ securely realizes ${{\hat Y}_{\mathcal{A}}}$ using a setup ${X_\mathcal{A}}$, $\mathfrak{T}({X_{\mathcal{A}}})$ should be contained within $\mathfrak{T}({\hat Y}_{\mathcal{A}})$.
Definition 3 makes this precise.

\begin{definition}
Let $\mathcal{M}$ be a function that maps the $K$-tuple of RVs ${X_\mathcal{A}}$ to a subset of $\mathbb{R}_ + ^d$ s.t. if $a \in \mathcal{M}$ and $a' \geq a$, then $a' \in \mathcal{M}$. 
$\mathcal{M}$ is a monotone region if the following hold:

\begin{trivlist}
  \item \emph{1) Monotonicity under local operations (LO):} Suppose party-$i$ modifies $X_i$ to $Z$ by sending $X_i$ over a channel, characterized by $p_{Z|{X_i}}$. Then $\mathcal{M}$ cannot shrink, i.e., for all jointly distributed RVs $({X_\mathcal{A}},Z)$ with ${X_{\mathcal{A}\backslash i}} - {X_i} - Z$,
  $\mathcal{M}(X_1;\ldots;X_{i}Z;\ldots;X_K) \supseteq \mathcal{M}(X_1;\ldots;X_{i};\ldots;X_K)$.
   \item \emph{2) Monotonicity under public communication (PC):} Suppose party-$i$ publicly announces the value of $\widetilde X_i$. Then $\mathcal{M}$ cannot shrink, i.e., for all jointly distributed RVs $({X_\mathcal{A}},\widetilde X_i)$ with $H(\widetilde X_i|X_i)=0$, $\mathcal{M}(\widetilde X_iX_1;\ldots;\widetilde X_iX_{i-1};X_{i};\widetilde X_iX_{i+1};\ldots;\widetilde X_iX_K)$ $\supseteq$ $\mathcal{M}(X_1;\ldots;X_{i};\ldots;X_K)$.
  \item \emph{3) Monotonicity under statistically secure sampling:} Suppose, a subset $\mathcal{T}$ of the parties are ``passively corrupted'' who retain and share their views (encapsulated in the RV $V_{\mathcal{T}}$) in an attempt to infer additional information on the outputs of the non-coalition parties. 
W.l.o.g. let $\mathcal{T}=\{1,\ldots,m\}$, where $m \leq t$. For all jointly distributed RVs $({\hat Y_\mathcal{A}},{V_\mathcal{T}})$ and ${\delta _\mathcal{T}} \geq 0$, for each such $\mathcal{T}$ $(\subset\mathcal{A})$ if $I({V_\mathcal{T}};{\hat Y_{\mathcal{A}\backslash \mathcal{T}}}|{\hat Y_\mathcal{T}}) \leq {\delta _\mathcal{T}}$, 
then $\mathcal{M}({\hat Y_{1}};\ldots;{\hat Y_{m}};{\hat Y_{m+1}};\ldots;{\hat Y_{K}}) \supseteq \mathcal{M}({\hat Y_{1}}V_1;\ldots;{\hat Y_{m}}V_m;{\hat Y_{m+1}};\ldots;{\hat Y_{K}})+{\delta _\mathcal{T}}$, 
i.e., statistically securely sampled outputs do not have a much smaller region.
\item \emph{4) Additivity:} $\mathcal{M}$ supports coordinate-wise Minkowski addition for tensor products and is superadditive in general.
\item \emph{5) Continuity, Convexity and Closure:} $\mathcal{M}$ is a continuous function of the joint pmf $p_{X_\mathcal{A}}$. Also $\mathcal{M}$ is convex and closed.
\end{trivlist}
\end{definition}

\begin{theorem}
$\mathfrak{T}$ is a $(K+1)$-dimensional monotone region.
\end{theorem}

\begin{proof}
The following monotonicity inequality is useful: $I(X;Y|f(X)Z) \leq I(X;Y|Z)$.

\vspace{-.5\baselineskip}
\begin{trivlist}
  \item 1) For the joint pmf $p_{{X_\mathcal{A}}ZQ} = {p_{{X_\mathcal{A}}}}{p_{Z|{X_i}}}{p_{Q|{X_\mathcal{A}}}}$, monotonicity under LO holds since,
  \[\begin{aligned}
  {\Delta _{2i}}&:I({X_{\mathcal{A}\backslash i}};Q|{X_i}Z) = I({X_{\mathcal{A}\backslash i}};Q|{X_i}), \hfill \\
  \mathop {{\Delta _{2j}}}\limits_{j \ne i} &:I({X_{\mathcal{A}\backslash j}}Z;Q|{X_j}) = I({X_{\mathcal{A}\backslash j}};Q|{X_j}), \hfill \\
  \mathop {{\Delta _1}}\limits_{i = K} &:I({X_1}; \ldots ;{X_{K - 1}};{X_K}Z|Q)\mathop  = \limits^{{\text{(a)}}} I({X_1}; \ldots ;{X_K}|Q), \hfill \\ 
\end{aligned} \]
where (a) follows from choosing $i=K$ and using the recurrence relation $\Delta _1^K({X_1}; \ldots ;{X_K}|Q)$$\mathop=\Delta _1^{K - 1}({X_1};\ldots ;{X_{K - 1}}|Q)+I({X_K};{X_1} \ldots {X_{K - 1}}|Q)$. Since $\Delta _1$ is symmetric in all $X_i$'s, this holds for all parties.
  \item 2) For the joint pmf $p_{{X_\mathcal{A}}\widetilde X_iQ} = {p_{{X_\mathcal{A}}}}{p_{\widetilde X_i|{X_i}}}{p_{Q|{X_\mathcal{A}}}}$, monotonicity under PC holds since,
  \[\begin{aligned}
  {\Delta _{2i}}&:I({X_{\mathcal{A}\backslash i}}\widetilde X_i;Q\widetilde X_i|{X_i}) = I({X_{\mathcal{A}\backslash i}};Q|{X_i}), \hfill \\
  \mathop {{\Delta _{2j}}}\limits_{j \ne i} &:I({X_{\mathcal{A}\backslash j}}\widetilde X_i;Q\widetilde X_i|{X_j}\widetilde X_i) \leq I({X_{\mathcal{A}\backslash j}};Q|{X_j}), \hfill \\
  \mathop {{\Delta _1}}\limits_{i = 1} &:I({X_1};\widetilde X_1X_2; \ldots ;{\widetilde X_1X_K}|\widetilde X_1Q) \hfill \\
  &= {I({X_1};\widetilde X_1X_2|\widetilde X_1Q)}+\sum\limits_{j = 2}^{K - 1} {I(\widetilde X_1{X_1} \ldots {X_j};{\widetilde X_1X_{j + 1}}|\widetilde X_1Q)}  \hfill \\
  &\leq {I({X_1};X_2|Q)}+\sum\nolimits_{j = 2}^{K - 1} {I({X_1} \ldots {X_j};{X_{j + 1}}|Q)}  \hfill \\
  &= {I(X_1;\ldots;X_K|Q)}, \hfill \\ \end{aligned}\]
  where we have chosen $i=1$. Since $\Delta _1$ is a symmetric quantity, this holds for all $i$.
  \item 3) For any ${p_{Q|{{\hat Y}_\mathcal{A}}{V_\mathcal{T}}}} \in {\mathcal{P}_{{{\hat Y}_\mathcal{A}}{V_{\mathcal{T}}}}}$, monotonicity under statistically secure sampling easily holds for $\Delta_1$. For the coordinates ${{\{ {\Delta _{2i}}\} }_{i \in \mathcal{T}}}$, if $I({V_\mathcal{T}};{\hat Y_{\mathcal{A}\backslash \mathcal{T}}}|{\hat Y_\mathcal{T}}) \leq {\delta _\mathcal{T}}$, we have 
  \[\begin{aligned}
  \mathop {{\Delta _{2i}}}\limits_{i \in \mathcal{T}} &:I({V_{\mathcal{T}\backslash i}}{{\hat Y}_{\mathcal{A}\backslash i}};Q|{V_i}{{\hat Y}_i})=I({V_{\mathcal{T}\backslash i}}{{\hat Y}_{\mathcal{T}\backslash i}}{{\hat Y}_{\mathcal{A}\backslash {\mathcal{T}}}};Q|{V_i}{{\hat Y}_i}) \hfill \\
   &= I({{\hat Y}_{\mathcal{A}\backslash {\mathcal{T}}}};Q|{V_{\mathcal{T}}}{{\hat Y}_{\mathcal{T}}}) + I({V_{\mathcal{T}\backslash i}}{{\hat Y}_{\mathcal{T}\backslash i}};Q|{V_i}{{\hat Y}_i}) \hfill \\
   &\geq I({{\hat Y}_{\mathcal{A}\backslash {\mathcal{T}}}};Q|{V_{\mathcal{T}}}{{\hat Y}_{\mathcal{T}}}) \hfill \\
   &= I({{\hat Y}_{\mathcal{A}\backslash {\mathcal{T}}}};Q{V_{\mathcal{T}}}|{{\hat Y}_{\mathcal{T}}}) - I({V_{\mathcal{T}}};{{\hat Y}_{\mathcal{A}\backslash {\mathcal{T}}}}|{{\hat Y}_{\mathcal{T}}}) \hfill \\
   &\geq I({{\hat Y}_{\mathcal{A}\backslash {\mathcal{T}}}};Q|{{\hat Y}_{\mathcal{T}}}) - I({V_{\mathcal{T}}};{{\hat Y}_{\mathcal{A}\backslash {\mathcal{T}}}}|{{\hat Y}_{\mathcal{T}}}) \hfill \\
   &\Rightarrow I({{\hat Y}_{\mathcal{A}\backslash {\mathcal{T}}}};Q|{{\hat Y}_{\mathcal{T}}}) \leq I({V_{\mathcal{T}\backslash i}}{{\hat Y}_{\mathcal{A}\backslash i}};Q|{V_i}{{\hat Y}_i}) + {\delta _{\mathcal{T}}}. \hfill \\
\end{aligned}\]
For ${{{\{{\Delta _{2j}}\}}_{j \notin \mathcal{T}}}}$, $I({{\hat Y}_{\mathcal{A}\backslash j}};Q|{{\hat Y}_j})$ $\leq$ $I({V_{\mathcal{T}}}{{\hat Y}_{\mathcal{A}\backslash j}};Q|{{\hat Y}_j})$ $+$ ${\delta _{\mathcal{T}}}$.
\item 4) Additivity on tensor products and more generally superadditivity follows using arguments very similar to the ones for the $K=2$ case [2]. 
\item 5) Continuity and closure follow from Theorem 2. Convexity follows from arguments similar to the $K=2$ case (see Theorem 2.4 and 2.5 in [2]).   
\end{trivlist}
\end{proof}
\vspace{1mm}

Our generalization yields an interesting quantity (see Theorem 4), $\Delta _{1}^{\operatorname{int}}({{Y}_1}; \ldots ;{{Y}_K})$ which we call the \emph{residual total correlation}. Total correlation, $I({Y_1}; \ldots ;{Y_K})$ is a natural generalization of the mutual information in the multipartite case [5] that admits a simple operational interpretation: if parties in distant labs who share a noisy correlation $(p_{Y_{\mathcal{A}}})$ choose to \emph{forget} all correlations between them by locally processing $Y_i$ in their labs (e.g., sending $Y_i$ through a channel that completely randomizes it), then total correlation is the minimum increase of entropy of the local uncorrelated labs. Total correlation is a monotone [5] as is its residual counterpart. The latter follows from Theorem 4 and Theorem 5 since $\Delta _{1}^{\operatorname{int}}$ is the GK axis intercept of the boundary of $\mathfrak{T}({Y_{\mathcal{A}}})$ that measures the gap between total correlation and GK CI (see (4)). Condition (3) in Definition 3 implies (among other things), the following data processing inequality for $\Delta _{1}^{\operatorname{int}}$: the residual total correlation can never increase under any secure mapping from views to outputs. Analogous to the case for $K=2$ [1], we can state the following result for $t=1$, the weakest form of \emph{t-privacy}.

\vspace{.5\baselineskip}
\begin{proposition}
For all jointly distributed RVs $({Y_\mathcal{A}},{V_\mathcal{A}})$, if ${V_i}-{ Y_i}-{ Y_{\mathcal{A}\backslash i}}$, then $\Delta _{1}^{\operatorname{int}}({{Y}_1}; \ldots ;{{Y}_K})$ $\leq$ $\Delta _{1}^{\operatorname{int}}({V_1{Y}_1}; \ldots ;V_K{{Y}_K})$.
\end{proposition}

The most important consequence of Theorem 5 is that $\mathfrak{T}$ can be used to derive the impossibility of sampling ${Y_{\mathcal{A}}}$ from ${X_{\mathcal{A}}}$ with $\epsilon$\emph{-correctness} and ($\delta$,\emph{t})-\emph{privacy}---unless and until $\mathfrak{T}(X_{\mathcal{A}})$ $\subseteq$ $\mathfrak{T}({Y_{\mathcal{A}}})$, such reductions are impossible. Furthermore, by virtue of the continuity and convexity of $\mathfrak{T}$, one can derive an upper bound on the rate of such reductions. We prove a milder version of the above statement in Corollary 7. An analogous statement for the rate requires invoking arguments related to the convexity of the monotone region which we skip. The details are similar to the argument in [2].
\vspace{.5\baselineskip}
\begin{corollary}
If $m$ i.i.d copies of $Y_{\mathcal{A}}$ can be statistically securely realized from $n$ i.i.d copies of $X_{\mathcal{A}}$, then $n\mathfrak{T}(X_{\mathcal{A}}) \subseteq m\mathfrak{T}({Y_{\mathcal{A}}})$, (where multiplication by $n$ refers to $n$-times repeated Minkowski sum).
\end{corollary}

\begin{proof}[Proof (sketch)]
Let the RV $V_{\mathcal{A}}^r$ encapsulate the view of the parties at the end of round $r$. Let $V_{\mathcal{A}}^0 = X_{\mathcal{A}}^n$ and let the final view be $V_{\mathcal{A}}$. Then the proof follows from Theorem 5 by noting the following. By Condition (1) and (2) of Definition 3, ${\mathfrak{T}}(V_{\mathcal{A}}^r) \supseteq {\mathfrak{T}}(V_{\mathcal{A}}^{r - 1})$. By Condition (3), ${\mathfrak{T}}(Y_{\mathcal{A}}^m) \supseteq {\mathfrak{T}}({V_{\mathcal{A}}})$. Thus, ${\mathfrak{T}}(Y_{\mathcal{A}}^m) \supseteq {\mathfrak{T}}(X_{\mathcal{A}}^n)$. Finally, by Condition (4), the required inclusion holds.
\end{proof}

Given ${p_{Q|{Y_{\mathcal{A}}}}} \in {{\mathcal{P}}_{{Y_{\mathcal{A}}}}}$, the set of all \emph{t-private} distributions that can be sampled from scratch with perfect correctness and privacy, are characterized by the following conditions:
\begin{align*}
  &{\Delta _{2i}} = I({Y_{\mathcal{A}\backslash i}};Q|{Y_i}) = 0,{\text{ }}\forall i \in \mathcal{A} \hfill \tag{6}\\
  &{\Delta _1} = I({Y_1}; \ldots; {Y_K}|Q) = \sum\nolimits_{i = 1}^{K - 1} {I({Y_1} \ldots {Y_i};{Y_{i + 1}}|Q)}  = 0 \tag{7}\hfill
\end{align*}
\emph{t-privacy} follows from (6), (7) since  ${\Delta _{2i}} = I({Y_{\mathcal{A}\backslash i}};Q|{Y_i}) = 0,{\text{ }}\forall i \in \mathcal{A}$ $\Rightarrow$ $I({Y_{{\mathcal{A}\backslash \mathcal{T}}}};Q|{Y_{\mathcal{T}}}) = 0,{\text{ }}\forall {\mathcal{T}} \subset {\mathcal{A}},{\text{ }}\lvert{\mathcal{T}}\rvert \leq t$, and ${\Delta _1} = \sum\nolimits_{i = 1}^{K - 1} {I({Y_1} \ldots {Y_i};{Y_{i + 1}}|Q)} = 0$ $\Rightarrow$  $I({Y_{{\mathcal{A}\backslash }i}};{Y_i}|Q) = 0,{\text{ }}\forall a \in {\mathcal{A}}$ $\Rightarrow$ $I({Y_{{\mathcal{A}\backslash \mathcal{T}}}};{Y_{\mathcal{T}}}|Q) = 0,{\text{ }}\forall {\mathcal{T}} \subset {\mathcal{A}},{\text{ }}\lvert{\mathcal{T}}\rvert \leq t$.

A $2K$-dimensional characterization for the $K$-variate monotone region, ${\mathfrak{T}}^{\mathsf{2K}}$ was given in [3] (see Theorem 3 in [3]), by further decomposing the residual total dependency, $\Delta_1$ into $K$ components, viz.,
\begin{displaymath}
{\mathfrak{T}}^{\mathsf{2K}}({Y_{\mathcal{A}}}) = \begin{cases}
(\{{R_{i_1}}\} _{i_1 = 1}^K,\{{R_{i_2}}\} _{i_2 = 1}^K):\exists {p_{Q|{X_\mathcal{A}}}} \text{ s.t.} \forall i \in \mathcal{A},\\
{R_{i_1}} \geq I({Y_{\mathcal{A}\backslash i}};{Y_i}|Q), {R_{i_2}} \geq I({Y_{\mathcal{A}\backslash i}};Q|{Y_i}). \\
\end{cases}
\end{displaymath}
For independent setups, both ${\mathfrak{T}}^{\mathsf{2K}}$ and ${\mathfrak{T}}$ yield the same characterization of the \emph{t-private} distributions realizable from scratch. With non-trivial setups $(K > 2)$, ${\mathfrak{T}}$ can give strictly tighter bounds (than ${\mathfrak{T}}^{\mathsf{2K}}$) on the rates of secure $K$-party protocols. This follows from noting that whenever the common core $Q$ fails to completely resolve the dependence between $Y_{\mathcal{A}}$, any decomposition of $\Delta_1$ of the form $\sum\nolimits_{i = 1}^K {I({Y_{{\mathcal{A}}\backslash i}};{Y_i}|Q)}$ is bound to induce some redundant mutual information terms. Theorem 8 gives sufficient conditions for the statistical case.

\vspace{.5\baselineskip}
\begin{theorem}
A $K$-tuple of RVs ${Y}_{\mathcal{A}}\sim{p_{{Y}_{\mathcal{A}}}}$ can be sampled from scratch with $\epsilon$-correctness and ($\delta$,t)-privacy, if there exists a RV $Q$, jointly distributed with ${\hat Y}_{\mathcal{A}}$ s.t. the following hold:
\begin{align*}
\mathsf{TV}(p_{{Y_\mathcal{A}}},p_{{{\hat Y}_\mathcal{A}}}) \leq \epsilon \tag{8}\\
\sum\limits_{\mathcal{T} \subset \mathcal{A}:\lvert{\mathcal{T}}\rvert \leq t} {{I({\hat Y}_{\mathcal{A}\backslash \mathcal{T}};Q|{\hat Y}_\mathcal{T})}  \leq \delta} \tag{9} \\
I({\hat Y}_{{\mathcal{A}}\backslash i};{\hat Y}_i|Q) = 0, \hspace{2mm}\forall a \in \mathcal{A} \tag{10}
\end{align*}
\end{theorem}

\begin{proof}
Consider the following protocol $\Pi _S$ satisfying conditions (8)--(10). Party-$i$ samples $U_i=({\hat Y}_i, Q)$ and publicly discloses the value of $Q$, following which, each ${\{ {\text{party-}}j\} _{j \in {\mathcal{A}}\backslash i}}$ independently samples ${U}_j$ by flipping their private coins using $p_{{\hat Y_j}|Q}$ conditioned on the received $Q$. Then, from (10) it follows that ${\hat Y}_j$ are independent given $Q$ which implies ${\hat Y}_{\mathcal{A}}\sim{p_{{\hat Y}_{\mathcal{A}}}}$. Then, given (8), $\epsilon$-correctness follows. 

To show ($\delta$,\emph{t})-privacy, first note that $H({{\hat Y}_{{\mathcal{A}}\backslash {\mathcal{T}}}}|{{\hat Y}_{\mathcal{T}}}Q)$ $\mathop  = \limits^{{\text{(a)}}}$ $H({{\hat Y}_{{\mathcal{A}}\backslash {\mathcal{T}}}}|Q)$ $\mathop  = \limits^{{\text{(b)}}}$ $H({{\hat Y}_{{\mathcal{A}}\backslash {\mathcal{T}}}}|QU_{\mathcal{T}})$ $\mathop  = \limits^{{\text{(c)}}}$ $H({{\hat Y}_{{\mathcal{A}}\backslash {\mathcal{T}}}}|QU_{\mathcal{T}}{{\hat Y}_{\mathcal{T}}})$, where (a) follows from (10),
(b) follows from noting that $I({{\hat Y}_{{\mathcal{A}}\backslash {\mathcal{T}}}};U_{\mathcal{T}}|Q)=0$, and
(c) follows since ${{\hat Y}_{\mathcal{T}}}$ is a deterministic function of $(U_{\mathcal{T}},Q)$. Then
\[\begin{aligned}
I({{\hat Y}_{{\mathcal{A}}\backslash {\mathcal{T}}}};Q|{{\hat Y}_{\mathcal{T}}}) &= 
 H({{\hat Y}_{{\mathcal{A}}\backslash {\mathcal{T}}}}|{{\hat Y}_{\mathcal{T}}}) - H({{\hat Y}_{{\mathcal{A}}\backslash {\mathcal{T}}}}|{{\hat Y}_{\mathcal{T}}}Q) \\
 &=H({{\hat Y}_{{\mathcal{A}}\backslash {\mathcal{T}}}}|{{\hat Y}_{\mathcal{T}}}) - H({{\hat Y}_{{\mathcal{A}}\backslash {\mathcal{T}}}}|QU_{\mathcal{T}}{{\hat Y}_{\mathcal{T}}}) \\
 &=I({{\hat Y}_{{\mathcal{A}}\backslash {\mathcal{T}}}};Q U_{\mathcal{T}}|{{\hat Y}_{\mathcal{T}}})\\
 &\mathop  = \limits^{{\text{(d)}}}{I({{\hat Y}_{\mathcal{A}\backslash \mathcal{T}}};{V_\mathcal{T}}|{{\hat Y}_\mathcal{T}}) \mathop \leq \limits^{{\text{(e)}}} \delta},
\end{aligned} \]
where (d) follows since the view ${V_\mathcal{T}}$ comprises of private randomness $U_{\mathcal{T}}$ and $Q$, the sole message broadcast by party-$i$ at the start of the protocol, and (e) follows from (9). Then from Definition 2, it follows that $\Pi _S$ is ($\delta$,\emph{t})-private. 
\end{proof}

In [13], monotone region for a channel-type model ($K=2$) was defined under a restriction to the ${\Delta _{21}}=0$ plane to derive upper bounds on the oblivious transfer capacity. Equivalent generalizations for multiuser channels using pairwise setups are of interest. Another observation of independent interest is that recently, the Hypercontractivity (HC) ribbon, a tensorizing measure of correlation [14], was derived as a dual of the GW region [15]. Both the HC ribbon and ARI region behave monotonically under local stochastic evolution and are measures of nonlocal correlation. We leave as an open question as to how these regions might be related.

\vspace{-1mm}
\section*{Acknowledgment}
\vspace{1mm}
The author wishes to thank Paul Cuff and Manoj Prabhakaran for short useful discussions over email, as well as anonymous reviewers for their valuable comments.

\begin{IEEEbiography}[{\includegraphics[width=1in,height=1.25in,clip,keepaspectratio]{picture}}]{John Doe}
\blindtext
\end{IEEEbiography}

\vspace{-3mm}
\appendix
\vspace{1mm}
\begin{proof}[Proof (sketch) for Theorem 2]
The proof for achievability which is based on a generalized lossy source coding problem (for $K$ variables) follows similar lines as in [2] and is omitted in the interest of space. The converse follows by minor modifications from the $K=2$ case [2] and is provided here for completeness.
\[\begin{aligned}
  n({R_a} + \epsilon) &\geq H({M_a}) \geq H({M_a}|X_a^n) \geq H({W_a}|X_a^n) \hfill \\
   &\geq I(Y_a^n;{W_a}|X_a^n),   \hspace{3mm} {Y_a} \triangleq X_{{\mathcal{A}}\backslash a}^{} \hfill \\
  &\mathop  =  \limits^{{\text{(a)}}} \sum\nolimits_{i = 1}^n {H(} Y_{ai}^{}|X_{ai}^{}) - H(Y_{ai}^{}|{W_a}Y_{ai}^{i - 1}X_a^n) \hfill \\
  &\geq \sum\nolimits_{i = 1}^n {H(} Y_{ai}^{}|X_{ai}^{}) - H(Y_{ai}^{}|{W_a}Y_{ai}^{i - 1}X_{ai}^{}X_a^{i - 1}) \hfill \\
   &= \sum\nolimits_{i = 1}^n {I(} Y_{ai}^{};{Q_i}|X_{ai}^{}), \hspace{3mm} {Q_i} \triangleq {W_a}X_{k + 1}^{i - 1} \ldots X_1^{i - 1} \hfill \\
  &\mathop  =  \limits^{{\text{(b)}}} \hspace{2mm} nI({Y_{aJ}};{Q_J}|{X_{aJ}}J), \hspace{3mm} {p_J}(i) \triangleq \tfrac{1}
{n},i \in \{ 1, \ldots ,n\} , \hfill \\
  &\mathop  =  \limits^{{\text{(c)}}} \hspace{2mm} nI({Y_{aJ}};Q|{X_{aJ}}), \hspace{3mm}Q \triangleq ({Q_J},J), \hfill \\
  \end{aligned} \]
where (a) follows from the independence of the $K$-tuple ${X_{\mathcal{A},i}} = {\{ {X_{a,i}}\} _{a \in \mathcal{A}}}$
across $i$. In (b), $J \in \{ 1, \ldots ,n\} $ is a uniformly distributed RV independent of $X_{\mathcal{A}}^n$ and (c) follows from the independence of $J$ and $X_{\mathcal{A}}^n$.
\begin{multline*}
I(X_1^n; \ldots ;X_K^n|{W_a}) = \sum\nolimits_{k = 1}^{K - 1} {I(X_{k + 1}^n;X_1^n \ldots X_k^n|{W_a})}\\
\shoveleft{= \sum\nolimits_{k = 1}^{K - 1} {\sum\nolimits_{i = 1}^n {I(X_{k + 1,i}^{};X_1^n \ldots X_k^n|{W_a}X_{k + 1}^{i - 1})} }} \\
\geq \sum\nolimits_{k = 1}^{K - 1} {\sum\nolimits_{i = 1}^n {I(X_{k + 1,i}^{};X_{1,i}^n \ldots X_{k,i}^n|{W_a}X_{k + 1}^{i - 1} \ldots X_1^{i - 1})} } \\
{= nI(X_{1J}^{}; \ldots ;X_{KJ}^{}|Q).}
\end{multline*}
The converse follows, since $(X_{1J}^{}, \ldots ,X_{KJ}^{})$ has the same distribution as $(X_1^{}, \ldots ,X_K^{})$. The cardinality bound on $Q$ can be shown using the Carathéodory-Fenchel theorem [12, p. 310]. The boundary of $\mathfrak{T}({X_{\mathcal{A}}})$ is thus made up of $(K+1)$-tuples of the form ${\Delta _m} = \left( {{{\{ {\Delta _{2a}}\} }_{a \in \mathcal{A}}},{\Delta _1}} \right)$, where ${\Delta _m}$ is a continuous function from ${\hat{\mathcal{P}}_{{X_\mathcal{A}}}} \to \mathbb{R}_ + ^{K + 1}$, where ${\hat{\mathcal{P}}_{{X_\mathcal{A}}}}$ is compact (i.e., closed and bounded). Since the image of a compact set under a continuous function is compact, $\{ {\Delta _m}:{p_{Q|{X_\mathcal{A}}}} \in {\hat{\mathcal{P}}_{{X_\mathcal{A}}}}\}$  is compact. Moreover, since the increasing hull of a compact set is closed (see Lemma A.3, [2]), $\mathfrak{T}$ is closed. Convexity of $\mathfrak{T}$ follows from arguments similar to the $K = 2$ case [2].         
\end{proof}

\begin{figure}[!t]
\centering
\includegraphics[width=1.2in]{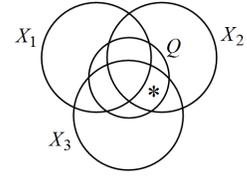}
\caption{Denoting the $I$-Measure of RV $Q$ by $\mu^*$, the only atom on which $\mu^*$ is nonvanishing is shown in the $I$-Diagram for the coordinate ${\Delta _{21}}$ on the boundary of $\mathfrak{T}({X_1};{X_2};{X_3})$}\vspace{-1em}
\label{fig_Idiag}
\end{figure}

\vspace{3mm}
\begin{proof}[Proof (sketch) for Theorem 4]
First note that $\mathfrak{T}({X_{\mathcal{A}}})$ intersects each of the $(K+1)$ axes, since any $K$-tuple of coordinates can be made simultaneously zero by choosing an appropriate Q. The case for $K = 2$ was already shown in [2]. For the intercept $\Delta _{21}^{\operatorname{int} }({X_1};{X_2};{X_3})$,
\[\begin{aligned}
  \Delta_{21}^{\operatorname{int}} &= \mathop {\inf }\limits_{\substack{
   I({X_3}{X_1};Q|{X_2}) = 0 \\ 
   I({X_1}{X_2};Q|{X_3}) = 0 \\ 
  I({X_1};{X_2}|Q) + I({X_1}{X_2};{X_3}|Q) = 0 
}} I({X_2}{X_3};Q|{X_1}) \hfill \\
  &\leq \mathop {\inf }\limits_{\substack{ 
   H(Q|{X_2}) = H(Q|{X_3}) = 0 \\ 
  I({X_1};{X_2}|Q) + I({X_1}{X_2};{X_3}|Q) = 0 
}}  {H(Q|{X_1})}, \hfill \\ 
\end{aligned} \]
since if $H(Q|{{X}_{2}})=H(Q|{{X}_{3}})=0,$ then $I({{X}_{3}}{{X}_{1}};Q|{{X}_{2}})=$ $I({{X}_{1}}{{X}_{2}};Q|{{X}_{3}})=0$ and $I({{X}_{2}}{{X}_{3}};Q|{{X}_{1}})=H(Q|{{X}_{1}}).$ For the converse, we want to show $\operatorname{LHS}\ge \operatorname{RHS}.$ This holds, since if $I({{X}_{3}}{{X}_{1}};Q|{{X}_{2}})=I({{X}_{1}}{{X}_{2}};Q|{{X}_{3}})=0$, then $H(Q|{{X}_{2}})$ $=H(Q|{{X}_{3}})=0$ and  $I({{X}_{2}}{{X}_{3}};Q|{{X}_{1}})=H(Q|{{X}_{1}})$. 

In fact, under the given constraints, denoting the $I$-Measure of RV $Q$ by $\mu^*$, the only atom on which $\mu^*$ is nonvanishing for both  $I({{X}_{2}}{{X}_{3}};Q|{{X}_{1}})$ and $H(Q|{{X}_{1}})$, is the one shown in the $I$-Diagram in Fig. 2. It may be noted that for $K=2$, the proof for the converse is not trivial (see Lemma A.1, A.2 and the proof of Theorem 2.2 in [2]), since given $I({{X}_{1}};Q|{{X}_{2}})=I({{X}_{1}};{{X}_{2}}|Q)=0,$ it does not trivially follow that $I({{X}_{2}};Q|{{X}_{1}}) \ge H(Q|{{X}_{1}})$. However, just as shown above for $K=3$, for $K\ge 3$ onwards, $\mu^*$ is vanishing on all but one atom, which trivially then yields the converse. Similar arguments hold for all the other coordinates and for any general $K$. Finally the use of $\min$ instead of $\inf$ in the statement of the theorem is valid since $\mathfrak{T}({X_{\mathcal{A}}})$ is closed.                                      
\end{proof}

\end{document}